\documentclass[a4paper,UKenglish]{lipics-v2016}

\usepackage{microtype}
\usepackage{todonotes}

\def\asssym{\textrm{ASS-S}}
\def\ass{\textrm{ASS}}
\def\rep{\textrm{REP}}
\def\pop{\textrm{POP}}
\def\pophybrid{\textrm{POP2}}
\def\minimize{\mathrm{min:}}
\def\subjectto{\mathrm{s.t.}}

\title{New Integer Linear Programming Models for the Vertex Coloring Problem\footnote{This work was partially supported by the German Research Foundation, RTG 1855.}}
\titlerunning{New ILP Models for Vertex Coloring} 

\author[1]{Adalat Jabrayilov}
\author[2]{Petra Mutzel}
\affil[1]{Department of Computer Science, TU Dortmund University, Germany\\
  \texttt{adalat.jabrayilov@cs.tu-dortmund.de}}
\affil[2]{Department of Computer Science, TU Dortmund University, Germany\\
  \texttt{petra.mutzel@cs.tu-dortmund.de}}
\authorrunning{A. Jabrayilov and P. Mutzel} 

\Copyright{Adalat Jabrayilov and Petra Mutzel}


\EventEditors{John Q. Open and Joan R. Acces}
\EventNoEds{2}
\EventLongTitle{42nd Conference on Very Important Topics (CVIT 2016)}
\EventShortTitle{CVIT 2016}
\EventAcronym{CVIT}
\EventYear{2016}
\EventDate{December 24--27, 2016}
\EventLocation{Little Whinging, United Kingdom}
\EventLogo{}
\SeriesVolume{42}
\ArticleNo{23}

\begin{document}

\maketitle

\begin{abstract}
The vertex coloring problem asks for the minimum number of colors that can be assigned to the vertices of a given graph such that for all vertices $v$ the color of $v$ is different from the color of any of its neighbors. 
The problem is NP-hard. 
Here, we introduce new integer linear programming formulations based on 
partial-ordering.
They have the advantage that they are as simple to work with as the classical assignment formulation,
since they can be fed directly into a standard integer linear programming solver.
We evaluate our new models using Gurobi and show that our new simple approach 
is a good alternative to the best state-of-the-art approaches for the vertex coloring problem.
In our computational experiments, we compare our formulations with the
classical assignment formulation and the representatives formulation
on a large set of benchmark graphs as well as randomly generated graphs of varying size and density.
The evaluation shows that the partial-ordering based models dominate both formulations for sparse graphs,
while the representatives formulation is the best for dense graphs.
   \end{abstract}

\section{Introduction}
Graph coloring belongs to the classical optimization problems and has been studied for a long time.
We consider the \emph{vertex coloring problem} in which colors are assigned to the vertices of a graph such that no two adjacent vertices get the same color and the number of colors is minimized.
The minimum number of colors needed for a given graph is called its \emph{chromatic number} and 
denoted by $\chi$. 
Computing the chromatic number of a graph is NP-hard \cite{GarJoh79}.
Since the graph coloring problem has many applications, e.g., register allocation, scheduling, frequency assignment and timetabling, there is a vast amount of literature on this problem (see, e.g., \cite{Malaguti2010} for a survey).
However, in contrast to other classical combinatorial optimization problems such as the Travelling Salesman Problem, where large instances can be solved to optimality, this is not true for the vertex coloring problem. 
So far only relatively small vertex coloring instances 
can be solved 
to provable optimality by exact algorithms.

There are three main directions followed by exact algorithms, two of which are based on integer linear programming (ILP) formulations for the problem. The natural formulation introduces binary variables that assign colors to vertices.
A vertex $v$ in the graph $(V,E)$ is assigned color $i$ iff the corresponding binary variable $x_{vi}$ gets value 1. 
This \emph{assignment formulation} has the advantage that it is simple and easy to use. 
Since the number of variables and constraints of this ILP model is 
polynomial
in $|V|$,
it can be fed directly into a (commercial) integer linear programming solver such as \emph{SCIP} \cite{SCIP2008}, \emph{LP\_Solve} \cite{LPSOLVE}, \emph{Cplex}, or \emph{Gurobi}.
Due to the inherent symmetry in the model (the colors are not distinguishable) only small instances can be solved to optimality.
However, additional constraints can be added by which the symmetry can be reduced.
 The second approach has been suggested by Mehrotra and Trick \cite{MehTri96} and is based on the observation that each color class defines an independent set (no two vertices in the set are adjacent) in the graph. The variables correspond to independent sets and the ILP model searches for the minimum number of these independent sets that cover the graph.
 Since the number of independent sets can be of exponential size, the solution approach is based on column generation.
Solution algorithms based on this \emph{Set Covering Formulation} are of complex nature. 
Mehrotra and Trick \cite{MehTri96} suggest a branch-and-price algorithm for solving the ILP model.
Both formulations have been studied and improved by additional ideas in the literature leading to complex branch-and-cut algorithms using additional classes of constraints, special branching schemes, separation procedures, and special procedures for providing good upper bounds.
The computational studies in the literature show that none of the ILP models dominates the other one.


ILP formulations based on 
partial-ordering
have shown to be practically successful in the area of graph drawing \cite{Jabrayilov2016}.
Here, we suggest a new ILP formulation based on 
partial-ordering
for the vertex coloring problem.
It has the advantage that it is as simple to work with as the natural assignment formulation,
since it is of polynomial size and can be fed directly into a standard integer linear programming solver.
We further suggest a hybrid ILP formulation which combines the advantages of the assignment formulation with those of the 
partial-ordering model.

We evaluate the new models using the ILP solver \emph{Gurobi} and show that our new simple approaches 
dominate the assignment formulation on the tested benchmark sets and are a good alternative to the best state-of-the-art approaches for the vertex coloring problem.
We also present the first experimental comparison with the representatives formulation which has been suggested by Campelo et al. \cite{Campelo2004,Campelo2008}. Since it introduces variables for every pair of non-adjacent vertices, this formulation seems to be advantageous for dense graphs. Our computational results support this observation.

\section{State-of-the-Art}\label{se:state-of-the-art}

Eppstein \cite{Epp03} has shown that it is possible to solve the vertex coloring problem by enumerating all maximal independent sets in the graph in time $O((4/3+3^{4/3}/4)^{|V|})$ which is about $2.4150^{|V|}$.
In practice, the successful approaches are much faster than that.
There are three main lines of research concerning exact solution methods for the coloring problem:
branch-and-bound approaches based on enumeration (Section \ref{sse:BaB}), the ILP-based assignment model (Section \ref{sse:ASS}), and the ILP-based set covering formulation (Section  \ref{sse:SC}).
From those, the assignment model is the simplest one, since it can directly be fed into a standard ILP-solver. 
There are also a lot of experimental evaluations of these methods.
Altogether they have not shown a superiority of one of the three lines of research.
Another simple ILP formulation is the so-called representatives ILP model (see Section \ref{sse:REP}).
It seems that there is no experimental evaluation concerning this model.
In the following we describe the state-of-the-art concerning these exact approaches.
In the literature there also exist alternative ILP models and alternative approaches (e.g., based on Constraint Programming \cite{gualandi2012exact}). 
However, the aim of this work is to concentrate on simple ILP formulations that are competitive with the best state-of-the-art approaches. 
There is also a vast literature on heuristic approaches.
For a detailed overview of heuristic and exact approaches, see the survey by Malaguti and Toth \cite{Malaguti2010} and Burke et al. \cite{Burke2010}.

\subsection{Branch-and-bound based approaches}\label{sse:BaB}
The first line is based on branch-and-bound.
Brown \cite{Brown1972} suggested a backtracking algorithm that colors the vertices iteratively.
He developes some methods in order to reduce redundant partial solutions.
Br\'elaz \cite{Brelaz1979} has further developed the ideas of the Randall-Brown algorithm.
He suggested a heuristic based on the saturation degree of a vertex (the number of different colors to which a vertex is adjacent in a partially colored graph). His heuristic, called \emph{Dsatur algorithm}, is exact for bipartite graphs.
For solving the VCP for general graphs, he suggested to start with an initial coloring of a large clique and a preprocessing step of using his Dsatur heuristic. This approach is called \emph{DSATUR-based branch-and-bound} in the literature or in short \emph{DSATUR}.
Sewell \cite{Sew96} has further improved the approach by introducing a new tie breaking strategy.
Recently, Segundo \cite{Segundo2012} has suggested to apply this strategy only selectively thus improving the overall performance. He studied the behaviour of his new algorithm \emph{PASS} on a subset of the DIMACS color benchmark and random graphs.

\subsection{Assignment-based ILP model}\label{sse:ASS}
The classical ILP model for a graph $G=(V,E)$ is based on assigning color $i$ to vertex $v\in V$. 
For this, it introduces assignment variables $x_{vi}$ for each vertex $v$ and color $i$ ($i=1,\ldots,H$), with 
$x_{vi}=1$ if vertex $v$ is assigned to color $i$ and $x_{vi}=0$ otherwise. 
$H$ is an upper bound of the number of colors (e.g., the result of a heuristic)
and is at most $|V|$.
For modelling the objective function, additional binary variables $w_i$ are needed which get value 1 if and only if color $i$ is used in the coloring ($i=1,\ldots, H$).
The model is given by:
%
%
%
\begin{eqnarray}
(\asssym)  &  \min \sum_{1 \le i \le H} w_{i} &     \label{vcp:ass} \\
 \subjectto\    &  \sum_{i=1}^H x_{vi}   =  1	& \forall v \in V \label{unambiguous}\\
&     x_{ui} + x_{vi}   \le  w_i	& \forall (u,v) \in E,\ i=1,\ldots,H \label{edge:coloring}\\
&    x_{vi}, w_{i}    \in   \{0,1\}	& \forall v \in V,\ i =1,\ldots, H 
\end{eqnarray}

  The objective function minimizes the number of used colors. Equation (\ref{unambiguous})
  ensures that each vertex receives exactly one color. For each edge there is a constraint (\ref{edge:coloring}) making sure that adjacent vertices receive different colors.
  This model has the advantage that it is simple and easy to use. It can be easily extended to generalizations and/or restricted variants of the graph coloring problem. 
  Since the number of variables is quadratic in $|V|$
  (it is bounded by $H(|V|+1)$) and the number of 
  constraints is cubic in $|V|$ (exactly $|V|+H|E|=O(|V||E|)$ constraints of type \ref{unambiguous} and \ref{edge:coloring}), it can directly be used as input for a standard ILP solver.
  
Malaguti and Toth mention two drawbacks of this model. The first one is the inherent symmetry in the ILP model, since 
there are $\binom H {\chi}$ possibilities
to select $\chi$ from $H$ colors
thus leading to exponentially (in the number of colors) 
many equivalent solutions.
Moreover, the continuous LP-relaxation is extremely weak, since it has a feasible solution of value 2 independent of the graph. It is possible to set the values of all vertices $v\in V$ as follows: $x_{v1}=x_{v2}=0.5$ and $x_{vj}=0$ for  $j=3,\ldots,H$, and $w_1=w_2=1$ and all other $w_i=0$.
%
%

In order to overcome the symmetry problem with the (\asssym) model, 
Mendez-Diaz and Zabala \cite{MenZab06,MenZab08} suggest 
to add the following additional set of constraints:
%
\begin{align}
    & w_i \le \sum_{v\in V} x_{vi}	&   i =1,\ldots, H\label{mz:cp1}\\
    & w_{i} \le w_{i-1}			&   i =2,\ldots, H  \label{mz:cp2}
\end{align}
%
%
These constraints ensure that the color $i$ is only assigned to some vertex, 
if color $i-1$ is already assigned to another one.
We call the extended model (\ass):
\begin{eqnarray*}
(\ass) &  \min \sum_{1 \le i \le H} w_{i}     \\
 \subjectto\    &  (\ref{unambiguous})-(\ref{mz:cp2})
\end{eqnarray*}
Moreover, they present several sets of constraints that arose from their studies of the polytope.
In order to solve the new strengthened ILP model, they developed a branch-and-cut algorithm.

\subsection{Representatives ILP model}\label{sse:REP}
A vertex coloring divides the vertices into disjoint color classes.
Campelo et al. \cite{Campelo2004,Campelo2008} suggested a model in which each color class is represented by exactly one vertex. For this, they suggest to introduce a binary variable $x_{uv}$ for each non-adjacent pair of vertices $u,v\in V$ which is 1 if and only if the color of $v$ is represented by $u$. Additional binary variables $x_{u,u}$ indicate if $u$ is the representative of its color class. Let $\bar N(u)$ be the set of all vertices not adjacent to vertex $u$. The constraints are as follows:
\begin{eqnarray}
(\rep) &  \min \sum_{u \in V} x_{uu} &     \label{vcp:rep} \\
&    \sum_{u\in \bar N(v) \cup v} x_{uv}   \ge 1 	&  \forall v\in V \label{eq:rep1} \\
&    x_{uv} + x_{uw}   \le x_{uu} 	& \forall u\in V,\ \forall  e=(v,w) \in G[\bar N(u)] \label{eq:rep2} \\
   & x_{uv} \in \{0,1\}			& \forall  \textrm{ non-adjacent vertex pairs } u,v \textrm{ or } u=v  \label{eq:rep3}
\end{eqnarray}

Inequality (\ref{eq:rep1}) requires that for any vertex $v\in V$, there must exist a color representative which can be $v$ itself or is in $\bar N(v)$. Inequality (\ref{eq:rep2}) states that a vertex $u$ cannot be the representative for both endpoints of an edge $(v,w)$ and that in the case that $x_{uv}=1$ ($u$ is the representative of vertex $v$) also the variable $x_{uu}$ must be 1. The advantage of this model is its simplicity and its compactness. 
It has exactly $|\bar E| + |V|$ variables and up to  $|V|+|V||E|$ many constraints, where $\bar E$ is the set of non-adjacent vertex pairs of $G=(V,E)$: $\bar E= (V \times V) \setminus E$. 
With growing density of the graphs, the number of constraints increases but the number of variables decreases and converges towards $|V|$.

In \cite{Campelo2008}, Campelo et al.\ mention that the symmetry in this model may lead to problems with branch-and-bound based solvers. The reason for this lies in the fact that within a color class any of the vertices in this class can be the representative. In order to circumvent this, the authors define an ordering on the vertices and require that in each color class only the vertex with the smallest number is allowed to be the representative of this class. However, the ILP model arising from this requirement, called the \emph{asymmetric representatives formulation}, has up to exponentially many constraints. So this model is no more simple to use. The authors \cite{Campelo2008} study the polytopes associated with both representative formulations. They suggest to add constraints based on cliques, odd-holes, and anti-holes, and wheels, and independent sets in order to strengthen the model. Moreover, they provide a comparison with the set covering based formulation (see Section \ref{sse:SC}).
In \cite{Campos2015}, Campos et al.\ study the asymmetric representatives formulation and the corresponding polytope. Their results lead to complete characterizations of the associated polytopes for some specific graph classes.
Up to our knowledge, no computational experiments have been published in the literature.

\subsection{Set covering based ILP model}\label{sse:SC}

  Mehrotra and Trick \cite{MehTri96} suggested the set-covering formulation which is based on the observation that the vertices receiving the same colors build an \emph{independent set}. A set of vertices is called \emph{independent set} if no two of its vertices are adjacent. 
  The aim of the formulation is to cover the vertices of the graph with the minimum number of independent sets. 
  Let $S$ a family of independent sets of the given graph $G=(V,E)$. 
 The ILP model  uses a binary variable $x_s$ for each independent set $s \in S$, with $x_s=1$ iff 
  the independent set $s$ is part of the cover. 
%
\begin{eqnarray}
\textrm{(COV)}  &  \min   \sum_{ s \in S} x_{s}   \label{vcp:sc}    \\
  \subjectto\ 
  & \sum_{s\in S:v \in s} x_s \ge 1	& \forall v \in V \label{eq-sc1} \\
    & x_s \in \{0,1\}			& \forall s \in S \label{eq-sc2}
\end{eqnarray}
%
%
  The objective function minimizes the number of independent sets used for covering the vertices.
  Constraint (\ref{eq-sc1}) ensures that each vertex is covered 
  by at least one independent set.
  Since the number of variables in this model can be of exponential size, this formulation cannot
  be fed directly into a standard ILP-solver like SCIP, LP\_Solve, Gurobi or CPLEX. 
  
  Mehrotra and Trick have suggested a branch-and-price algorithm for solving the model which starts with a small number of variables and adds additional ones using a column generation approach.
  Malaguti et al. \cite{Malaguti2011} have provided further details such as metaheuristic algorithms for initialization and column generation as well as new branching schemes.
  Held et al. \cite{Held2012} have suggested new techniques in order to improve the numerical stability of the branch-and-price method. 
  A very similar formulation, the \emph{Set Packing} formulation, has been suggested by Hansen et al. \cite{Hansen2009}.
Theoretical and empirical results show that both models have a similar behavior.
%





%
\section{Partial-Ordering Based ILP Models}

\subsection{A pure partial-ordering based ILP model: \pop}
\label{subsec:pop}

Our new binary model considers the vertex coloring problem as 
partial-ordering
problem (POP).
We assume that the $H$ colors ($1,\ldots,H$) are linearly ordered.
Instead of directly assigning a color to the vertices, 
we determine a partial order
of the union of the vertex set and the set of ordered colors. 
For this, we determine the relative order
of each vertex with respect to each color in the color ordering.
More specific: for every color $i$ and every vertex $v\in V$ our variables provide the information if $v$ is smaller or larger than $i$. 
We denote these relations by $v \prec i$ or $v \succ i$, resp.
In other words, the colors and the vertices build a partially ordered set in which all pairs of the form $(v,i)$ with $v\in V, i=1,\ldots,H$ are comparable.
%
%
%
%
We define the following POP variables:
\begin{gather*}
  \forall v \in V,\  i =1, \ldots, H: \hspace{12pt}
  y_{i,v} =
  \left\{
  \begin{array}[2]{ll}
    1 & \mbox{$v \succ i$}\\
      0 & \mbox{otherwise.} 
  \end{array}
  \right.
  \\ 
  \forall v \in V,\  i = 1, \ldots, H: \hspace{12pt}
  z_{v,i} =
  \left\{
  \begin{array}[2]{ll}
    1 & \mbox{$v \prec i$}\\
    0 & \mbox{otherwise.} 
  \end{array}
  \right.
\end{gather*}
If vertex $v$ has been assigned to color $i$, then $v$ is neither smaller nor larger than $i$ 
and we have $y_{i,v}=z_{v,i}=0$.
%
The connection with the assignment variables $x$ from the (ASS) model is as follows:
\begin{align}
  x_{vi} &= 1 - (y_{i,v}+z_{v,i}) & \forall v \in V,\  i = 1, \ldots, H
 \label{connection:pop-ass}
\end{align}
%
We select an arbitrary vertex $q \in V$ 
and formulate our new binary program as follows:
%
%
{
\def\hs{\hspace{8pt}} 
\begin{eqnarray}
(\pop) &  \min  1 + \sum_{1 \le i \le H} y_{i,q}    &   \\
  \subjectto\ 
    & z_{v,1} = 0  &\forall v \in V  \label{col:lrange} \\
    & y_{H,v} = 0   &\forall v \in V  \label{col:rrange}    \\
    & y_{i,v} - y_{i+1,v} \ge 0 & \forall v \in V,\  i =1,\ldots, H-1   
    \label{col:unambiguous1}   \\
    & y_{i,v} + z_{v,{i+1}} = 1 & \forall v \in V,\  i =1,\ldots, H-1   
    \label{col:unambiguous2} \\
    & y_{i,u} + z_{u,i} + y_{i,v} + z_{v,i} \ge 1 &\forall (u,v) \in E,\ i =1,\ldots, H  \label{col:edgecoloring}    \\
    & y_{i,q} - y_{i,v} \ge 0 & \forall v \in V,\  i =1,\ldots, H-1    \label{col:chi}     \\
    & y_{i,v}, z_{v,i} \in \{0,1\}  & \forall v \in V,\  i =1, \ldots, H
\end{eqnarray}
}
%
%
  \begin{lemma}
  The integer linear programming formulation (POP) described above is correct: Any feasible solution of the ILP corresponds to a feasible vertex coloring and the value of the objective function corresponds to the chromatic number of $G$.
  \end{lemma}
  
\begin{proof} 
All original vertices need to be embedded between the colors 1 and $H$. 
    Constraints (\ref{col:lrange}) and (\ref{col:rrange}) take care of this.
  By transitivity, a vertex that is larger than color  $i+1$ is also 
  larger than $i$ (constraints (\ref{col:unambiguous1})).
 Constraint  (\ref{col:unambiguous2}) expresses that each vertex $v$ is either
  larger than $i$ (i.e. $y_{i,v}=1$) or smaller than $i+1$ and not both.
  These constraints jointly with constraints (\ref{col:unambiguous1}) 
    ensure that each vertex will be assigned to exactly one color, i.e.
    there is no color pair $i \ne j$ with 
    $y_{i,v}=z_{v,i}=0$ and 
    $y_{j,v}=z_{v,j}=0$. 
    We show this by contradiction. 
    Let $y_{i,v}=z_{v,i}=0$. 
    %
	%
	In  the case $j<i$,  because of $z_{v,i}=0$ 
	we have 
	$y_{i-1,v}=1$ by (\ref{col:unambiguous2}).
	Therefore we have
	$y_{j,v}=1$ for each $j \le i-1$ by (\ref{col:unambiguous1})
	which is a contradiction to $y_{j,v}=0$.  
      In the case $j>i$, 
	because of $y_{i,v}=0$ we have 
	$y_{k,v}=0$ for each $k \ge i$ by (\ref{col:unambiguous1}).
	Therefore we have
	$z_{v,k+1}=1$
	by
	(\ref{col:unambiguous2}) leading to
	$z_{v,j}=1$ for each $j\ge i+1$ which is
	a contradiction to $z_{v,j}=0$.
    %
    %
    %
  Constraint (\ref{col:edgecoloring}) prevents assigning the same color $i$ to two adjacent 
  vertices $u$ and $v$.
  Constraint (\ref{col:chi}) takes care of the fact that our chosen vertex $q$ will be
  assigned to the largest chosen color. So if $q$ is not larger than color $i$ then
 this will be true for all other vertices $v \in V\setminus\{ q\}$.
    Because of this constraint, the objective function indeed minimizes the number
  of assigned colors since it sums up the number of
  colors smaller than $q$. In order to get the number of chosen colors, we need to add one  for the color assigned to $q$.
  We say that $q$ \emph{represents the chromatic number} of $G$.
\end{proof}

\paragraph*{Comparison with the assignment model}


(\pop) has $2 \cdot H|V|$ 
binary variables and about $4|V|H+2|V|+|E|$ constraints. 
Notice that 
the equations 
(\ref{col:lrange}),
(\ref{col:rrange}) and 
(\ref{col:unambiguous2})
can be used to eliminate
$(H+1)|V|$ variables. 
Hence the reduced model 
has $(H-1)\cdot|V|$ variables,
while the classical assignment model (\ass) has $H (|V|+1)$ variables.

Mendez-Diaz and Zabala \cite{MenZab06} mention that
the classical branching rule (to branch on fractional assignment variables by setting them to 1 in one subproblem and to 0
in another subproblem) produces quite unbalanced enumeration trees. 
This is the case because of setting $x_{v,i}=1$ implies $x_{v,j}=0$ for all $j\ne i$, while 
setting $x_{v,i}=0$ does not provide any further information.
The model (\pop) does not have this problem, since setting a \pop-variable
$y_{i,v}=0$ implies $y_{j,v}=0$ for all $j$ with  $j > i$
and setting $y_{i,v}=1$ implies $y_{j,v}=1$  for all $j$ with $j < i$
 because of constraint (\ref{col:unambiguous1}).


As already discussed, the original
(\asssym) model has inherent symmetries, which can be
resolved by additional constraints leading to the (\ass) model.
In the new (\pop) model, this type of symmetry does not occur.

Similarly as for the (ASS) model, the continuous LP-relaxation of the (POP) model is extremely weak, since it has a feasible solution of value 1.5 independent of the graph. It is possible to set the values of all vertices $v\in V$ as follows: 
$y_{1,v}=z_{v,2}=0.5$, and $y_{j,v}=0$ and $z_{v,k}=1$ for all other 
$j=2,\ldots,H$ and $k=3,\ldots,H$.

\subsection{A hybrid partial-ordering based ILP model: \pophybrid}\label{subsec:hybridMIP}
Our second ILP formulation is a slight modification of the first model and 
can be seen as a hybrid of the models (\pop) and (\ass).
It is the consequence of the observation that 
with growing density the (\pop) constraint matrix contains
more nonzero elements than the (\ass) constraint matrix.
This is due to the constraints (\ref{col:edgecoloring}), 
which are responsible for the valid coloring of 
adjacent vertices, and contain twice as many 
nonzero coefficients as the
corresponding (\ass) constraints (\ref{edge:coloring}).
Therefore, we use  equation (\ref{connection:pop-ass}) to substitute (\ref{col:edgecoloring}) by  (\ref{connection:pop-ass}) and the following constraints:
\begin{align}
  x_{ui} + x_{vi}   & \le  1  & \forall (u,v) \in E,\ i =1,\ldots, H.
  \label{pophybrid:edge:coloring}
\end{align}
This reduces the number of 
nonzero coefficients 
from $4 |E| H $ to
$2 |E| H + 3 |V|H $ giving a  reduction ratio of about two in dense
graphs.
Although we added $|V|H $ new assignment variables, the dimension of the
problem remains unchanged, since the new variables directly depend on the
\pop-variables by equality (\ref{connection:pop-ass}).





\section{Computational Experiments}

%
%
%

In our computational experiments we are interested in answering the following questions:
\begin{itemize}
\item (H1): Do our new 
partial-ordering
based ILP formulations dominate the classical assignment ILP model (ASS) on a set of benchmark instances? 
\item (H2): Does one of the two 
partial-ordering based
models dominate the other one?
\item (H3): How do the simple models behave compared to the 
state-of-the-art algorithms on a benchmark set of graphs?
\item (H4): Does the model (REP) dominate the other approaches on dense instances? 
\end{itemize}
%


\subsection{Algorithms and implementation issues}
\label{sse:preprocessing}


The preprocessing techniques (a)-(d) are widely used 
(e.g., 
\cite{gualandi2012exact,Hansen2009,Malaguti2010,MenZab06,MenZab08,Malaguti2011}):
%
\begin{enumerate}
  \item [(a)]
A vertex $u$ is \emph{dominated} by a vertex $v$, $v\not= u$, if the neighborhood of $u$ is a subset of the neighborhood of $v$. In this case, the vertex $u$ can be deleted from $G$, the remaining graph can be colored, and at the end, $u$ can get the same color as $v$.

  \item [(b)] To reduce the number of variables we are interested in getting a small upper bound $H$ for the number of needed colors.

  \item [(c)]Since any clique represents a valid lower bound for the vertex coloring problem
    one can select a clique and precolor it. This removes some variables, too.

  \item [(d)] In the case of equal lower and upper bounds 
    the optimal value has been found,
  hence no ILP needs to be solved. 

\end{enumerate}

We extended (c) as follows:
\begin{enumerate}
  \item [(e)] 
    In (\ass), (\pop), and (\pop2) 
    we can fix more variables if we try to find the clique $Q$ with
    $\max ( |Q|H + |\delta(Q)| )$, 
    where $\delta(Q):=\{(u,v) \in E \colon |\{u,v\} \cap Q| = 1\}$.
    The first term $|Q|H$ is due to the fact that we can fix $H$ variables for each vertex in
    $Q$.
    After precoloring all the 
    vertices $u \in Q$, the neighbors $v$ of $u$ cannot receive the same color
    as $u$. For example, if the assignment variable $x_{ui}=1$ then
    $x_{vi}=0$.
    Hence we can fix one variable for each edge $(u,v)\in\delta(Q)$.
\end{enumerate}

To represent the chromatic number in (\pop) and (\pophybrid),
we pick any vertex $q$ from the clique $Q$ found in the preprocessing.
The remaining vertices from the clique are precolored with colors $1,\cdots,|Q|-1$.

%
%

We have implemented the simple MIP models (ASS), (POP), (POP2), (REP) using the Gurobi-python API.
The source codes are available on our benchmark site \cite{POP:BENCHMARK}.
%
As already mentioned in subsection \ref{subsec:pop},
in our implementation of (\pop) and (\pophybrid)
we used the equations 
(\ref{col:lrange}),
(\ref{col:rrange}) and 
(\ref{col:unambiguous2})
and eliminated all $z$ variables. 
%

To compute (b), (c) and (e) we used the
python library \url{http://networkx.readthedocs.io}.
For (b) we used the function \texttt{networkx.greedy\_color()}.
To get 
a clique 
for (c) and (e) 
we applied the randomised function \texttt{networkx.maximal\_independent\_set()} 
on the complement graph, 
where the complement graph 
$G'$ of $G$ has the same vertices 
as $G$ but 
contains an edge $(u,v)$ for each vertex pair $u,v$ iff $G$ does not have 
this edge.
Since this function is randomised,
we iterated it 
at most $300 \cdot \frac{|E|}{|V|}$ times within a maximum time of 60 s
and selected the best one.



\subsection{Test setup and benchmark set of graphs}

For solving the ILP models, we used the Gurobi version 6.5 single-threadedly.
Due to the large benchmark set the experiments were performed on two computers:
 \begin{description}
   \item[M1] Intel Core i7-4790, 3.6GHz, with 32 GB of memory and running Ubuntu Linux 14.04. 
     (Benchmarks \cite{DIMACS:dfmax} user time: r500.5=4.14 s)

   \item[M2] Intel Xeon E5-2640, 2.60GHz, with 128 GB of memory and running Ubuntu Linux 14.04. 
     (Benchmarks \cite{DIMACS:dfmax} user time: r500.5=5.54 s)
 \end{description}



%
In the literature, subsets of the DIMACS 
\cite{DIMACS:BENCHMARK} 
benchmark sets and randomly generated graphs have been used.
We used also two sets, which are available at \cite{POP:BENCHMARK}.
From 119 DIMACS graphs we have chosen
the hard instances 
according to the Google benchmark site \cite{GOOGLE:BENCHMARK} 
and the 
\texttt{GPIA} graphs, which are obtained from a matrix partitioning problem to determine sparse Jacobian matrices.
The DIMACS instances are evaluated on computer M1.
  Furthermore, the new approaches are also compared 
  with the state-of-the-art algorithms \cite{MenZab06,MenZab08,Malaguti2011}
  which are of complex nature.
  They used the same DIMACS benchmark set and the same preprocessing techniques (a)-(d).
%
%

\begin{table}[ht]
{
\small
\scalebox{0.63}{
  \setlength{\tabcolsep}{2.0pt} 
  \begin{tabular}{l|rr|l|ccc|ccc|ccc|ccc|ccc|cc|c}
    \hline
      	         &      &        &       &           &\rep      &                &    &\pop      &               &         &\pophybrid   &                 &           &\ass+(c) &                 &           &\ass+(e) &                  &\cite{MenZab06}   &\cite{MenZab08}   &\cite{Malaguti2011}    \\
instance         &$|V|$ &$|E|$   &class  &lb         &ub        &time            &lb  &ub        &time           &lb       &ub           &time             &lb         &ub       &time             &lb         &ub       &time              &time  &time  &time     \\
\hline                                                                                                                                                                                              
1-FullIns\_4     &93    &593     &NP-m   &5          &5         &1.85            &5   &5         &0.01           &5        &5          &0.01             &5         &5         &0.01             &5        &5      &0.01              &0.1   &      &tl       \\
1-FullIns\_5     &282   &3247    &NP-?   &5          &6         &tl              &6   &6         &6.01           &6        &6          &1.54             &6         &6         &1.82		    &6        &6      &2.12		 &tl    &      &tl       \\
2-FullIns\_4     &212   &1621    &NP-m   &6          &6         &2.14            &6   &6         &0.02           &6        &6          &0.01             &6         &6         &0.01             &6        &6      &0.01              &tl    &4     &tl       \\
2-FullIns\_5     &852   &12201   &NP-?   &5          &7         &tl              &7   &7         &\textbf{5.02}  &7        &7          &72.45		 &7         &7         &326.61	    &7        &7      &14.74		 &tl    &      &tl       \\
3-FullIns\_3     &80    &346     &NP-m   &6          &6         &0.01            &6   &6         &0.00           &6        &6          &0.00             &6         &6         &0.00             &6        &6      &0.00              &0.1   &      &2.9      \\
3-FullIns\_4     &405   &3524    &NP-?   &7          &7         &3.92	         &7   &7         &0.03	         &7        &7          &0.02		 &7         &7         &0.02	       	    &7        &7      &0.03		 &tl    &      &tl       \\
3-FullIns\_5     &2030  &33751   &-      &7          &8         &tl              &8   &8         &\textbf{12.43} &8        &8          &32.90            &8         &8         &3489.97          &8        &8      &27.23             &tl    &      &tl       \\
4-FullIns\_3     &114   &541     &NP-m   &7          &7         &0.01            &7   &7         &0.00           &7        &7          &0.00             &7         &7         &0.00             &7        &7      &0.00              &3     &      &3.4      \\
4-FullIns\_4     &690   &6650    &NP-?   &8          &8         &4.61		 &8   &8         &0.05		 &8        &8          &0.02		 &8         &8         &0.03		    &8        &8      &0.02	         &tl    &      &tl       \\
4-FullIns\_5     &4146  &77305   &-      &7          &9         &tl              &9   &9         &\textbf{9.67}  &9        &9          &16.03            &8         &9         &tl               &9        &9      &78.17             &tl    &      &tl       \\
4-Insertions\_3  &79    &156     &NP-m   &3          &4         &tl              &4   &4         &\textbf{9.62}  &4        &4          &15.75            &4         &4         &141.48           &4        &4      &54.04             &4204  &      &tl       \\
5-FullIns\_3     &154   &792     &NP-m   &8          &8         &0.01            &8   &8         &0.00           &8        &8          &0.00             &8         &8         &0.00             &8        &8      &0.00              &20    &      &4.6      \\
5-FullIns\_4     &1085  &11395   &NP-?   &9          &9         &12.36		 &9   &9         &0.05		 &9        &9          &0.05		 &9         &9         &0.04		    &9        &9      &0.04		 &tl    &      &tl       \\
ash608GPIA       &1216  &7844    &NP-m   &-$\infty$  &1215      &tl              &4   &4         &\textbf{34.84} &4        &4          &51.63            &4         &4         &575.23           &4        &4      &821.74            &692   &      &2814.8   \\
ash958GPIA       &1916  &12506   &NP-m   &-$\infty$  &+$\infty$ &tl              &4   &4         &\textbf{90.11} &4        &4          &105.77           &4         &6         &tl               &4        &8      &tl                &tl    &4236  &tl       \\
DSJC125.5        &125   &3891    &NP-h   &14         &21        &tl              &11  &20        &tl             &13       &22         &tl               &13        &21        &tl               &13       &21     &tl                &tl    &      &18050.8  \\
DSJC125.9        &125   &6961    &NP-h   &44         &44        &\textbf{1.72}   &36  &50        &tl             &42       &44         &tl               &42        &45        &tl               &42       &45     &tl                &tl    &      &3896.9   \\
DSJR500.1c       &500   &121275  &NP-h   &85         &85        &\textbf{0.33}   &77  &+$\infty$ &tl             &83       &86         &tl               &78        &+$\infty$ &tl               &78       &+$\infty$        &tl                &tl    &      &288.5    \\
DSJR500.5        &500   &58862   &NP-h   &122        &497       &tl              &115 &+$\infty$ &tl             &122      &122        &\textbf{572.01}  &122       &122       &1748.11          &115      &+$\infty$        &tl                &tl    &      &342.2    \\
le450\_15a       &450   &8168    &NP-m   &15         &449       &tl              &15  &16        &tl             &15       &15         &\textbf{598.55}  &15        &15        &2439.49          &15       &15     &801.94            &tl    &      &0.4      \\
le450\_15b       &450   &8169    &NP-?   &15         &446       &tl              &15  &15        &2939.49	 &15       &15         &\textbf{700.50}  &15        &15        &1393.29	         &15       &15     &1103.29		 &tl    &      &0.2      \\
le450\_15c       &450   &16680   &NP-?   &-$\infty$  &450       &tl              &15  &+$\infty$ &tl             &15       &+$\infty$  &tl               &15        &+$\infty$ &tl               &15       &25     &tl                &tl    &      &3.1      \\
le450\_15d       &450   &16750   &NP-?   &-$\infty$  &450       &tl              &15  &26        &tl             &15       &26         &tl               &15        &+$\infty$ &tl               &15       &+$\infty$        &tl                &tl    &      &3.8      \\
le450\_25c       &450   &17343   &NP-?   &25         &450       &tl              &25  &30        &tl             &25       &31         &tl               &25        &+$\infty$ &tl               &25       &+$\infty$        &tl                &tl    &      &1356.6   \\
le450\_25d       &450   &17425   &NP-?   &25         &450       &tl              &25  &30        &tl             &25       &31         &tl               &25        &+$\infty$ &tl               &25       &+$\infty$        &tl                &tl    &      &66.6     \\
le450\_5a        &450   &5714    &NP-?   &-$\infty$  &450       &tl              &5   &9         &tl             &5        &5          &\textbf{21.17}   &5         &5         &83.65            &5        &5      &52.03    &tl    &      &0.3      \\
le450\_5b        &450   &5734    &NP-?   &-$\infty$  &450       &tl              &5   &7         &tl             &5        &5          &\textbf{140.16}  &5         &5         &503.29           &5        &5      &168.67   &tl    &      &0.2      \\
mug100\_1        &100   &166     &NP-m   &4          &4         &1.14            &4   &4         &0.24           &4        &4          &0.09             &4         &4         &0.39             &4        &4      &0.13              &60    &      &14.4     \\
mug100\_25       &100   &166     &NP-m   &4          &4         &1.10            &4   &4         &0.45           &4        &4          &0.31             &4         &4         &0.31             &4        &4      &0.31              &60    &      &12       \\
qg.order40       &1600  &62400   &NP-m   &-$\infty$  &+$\infty$ &tl              &40  &45        &tl             &40       &40         &\textbf{534.83}  &40        &46        &tl               &40       &46     &tl                &tl    &      &2.9      \\
qg.order60       &3600  &212400  &NP-?   &-$\infty$  &+$\infty$ &tl              &60  &68        &tl             &60       &62         &tl               &-$\infty$ &68        &tl               &-$\infty$  &68     &tl                &tl    &      &3.8      \\
queen10\_10      &100   &1470    &NP-h   &10         &12        &tl              &10  &12        &tl             &10       &12         &tl               &10        &12        &tl               &10       &12     &tl                &tl    &      &686.9    \\
queen11\_11      &121   &1980    &NP-h   &11         &13        &tl              &11  &13        &tl             &11       &13         &tl               &11        &13        &tl               &11       &13     &tl                &tl    &      &1865.7   \\
school1\_nsh     &352   &14612   &NP-m   &14         &14        &981.45          &14  &14        &22.39          &14       &14         &\textbf{12.76}   &14        &14        &31.23            &14       &14     &28.22             &0  &      &17       \\
wap05a           &905   &43081   &NP-m   &-$\infty$  &+$\infty$ &tl              &40  &+$\infty$ &tl             &50       &50         &1308.37          &50        &50        &\textbf{125.45}  &41       &+$\infty$        &tl                &tl    &      &293.2    \\
wap06a           &947   &43571   &NP-?   &-$\infty$  &+$\infty$ &tl              &40  &+$\infty$ &tl             &40       &+$\infty$  &tl		 &40        &+$\infty$ &tl               &40       &+$\infty$        &tl		     &tl    &      &175      \\
\hline                                                                                                                                                                                              
solved:          &      &        &       &           &       &13                 &      &       &19              &         &         &25         &                 &          &22        &                 &         &21                &9     &+2    &25       \\
 
    \hline
  \end{tabular}
}
}
  \caption{Results for the hard instances from DIMACS benchmark set}
  \label{tab:resultsDIMACS}
\end{table}
%
%
%
The second benchmark set consists of $340$ randomly generated
graphs $G(n,p)$, which have $n$ vertices and an edge between each vertex
pair with probability $p$. These graphs are 
evaluated on computer M2. This set consists of two subsets:
\begin{description}
  \item [set100:] This subset contains 100 instances with 70 vertices: 
	20 instances for each $p=0.1,0.3,0.5,0.7,0.9$.
  \item [sparse240:] This subset contains 240 instances, i.e.
    20 instances for each $n=80,90,100$ and for each $p=0.1,0.15,0.2,0.25$.
\end{description}


\subsection{Experimental evaluation}

Table \ref{tab:resultsDIMACS} shows the results for 
hard DIMACS benchmark instances 
for the ILP models 
(\pop), (\pophybrid), (\ass), (\rep) 
  and the state-of-the-art algorithms 
  \cite{MenZab06,MenZab08} and
  \cite{Malaguti2011}.
The table
contains 36 of the 68 hard instances, which can be solved by at least 
one of the considered algorithms.
Columns 1-3 show the instance names and sizes.
Column 4 describes the hardness
of the instances according to the Google site \cite{GOOGLE:BENCHMARK}.
Columns 5--19 display the lower and upper bounds as well as the 
running times of  the simple models (\rep), (\pop), (\pophybrid) and (\ass) 
that have been obtained within a time limit of one hour by the ILP-solver.  
An entry tl indicates that the time limit is reached.
The times are provided in seconds for solving the reduced ILPs after prepocessing. 
The preprocessing is done with python 
and took
only a few seconds for most instances and not more than one
minute. 
Columns 5--7 show the results of (\rep) with preprocessings (a)--(d).
The model (\rep) does not need (b) directly, but indirectly due to (d).
Columns 8--19 show appropriate results for (\pop), (\pophybrid) and (\ass).
Since (e) can reduce the number of assignment and POP variables, 
we implemented 
(\ass), (\pop)  and (\pophybrid) also with (e) instead of (c).
For (\pop) and (\pophybrid) it turned out that (e) is better than (c).
Due to space restrictions we present the results corresponding to (e) only.
For the (\ass) model we show the results for both versions (c) and (e).
The bold items in the table 
highlight the running times of the simple models that are significantly faster compared to the others.

Columns 20 and 21 are obtained from \cite{MenZab06,MenZab08}
and show the running times of two (\ass)-based branch-and-cut algorithms 
suggested by Mendez-Diaz and Zabala. The column 21 (\cite{MenZab08}) 
contains only the additional solved instances, i.e. the instances which have not been
solved in \cite{MenZab06} (column 20).
Column 22 is obtained from \cite{Malaguti2011} 
and shows the running times of a (COV)-based branch-and-price algorithm suggested 
by Malaguti et al. \cite{Malaguti2011}.
  Notice that
  the comparison of running times is not quite fair, 
  since \cite{MenZab08} and \cite{Malaguti2011} report  the Benchmark \cite{DIMACS:dfmax} user time for the instance r500.5 as
  24 s and 7 s, respectively, while our machine M1 needs 4.14 s.
  However,
  \cite{MenZab06,MenZab08} and \cite{Malaguti2011} used 2 h and 10 h as time limit
  respectively, while we only used 1 h.
  Nevertheless, it is interesting to see the number of solved instances by each algorithm.
  This is displayed in the last row of the table.
%
%
%
%

We can see that the hybrid model (POP2) and \cite{Malaguti2011} solved the highest number (25 out of 68) instances
to provable optimality.
Notice that (\pophybrid) used 1 hour as time limit, 
while \cite{Malaguti2011} needed more than 5 hours
for one instance
with a machine which is approximatively 1.7 times slower than M1.
The models (\ass)+(c), (\ass)+(e), (\pop), (\rep) have solved 22, 21, 19, 13 instances respectively,
while the algorithms \cite{MenZab06,MenZab08} only solved 11 (9+2) instances.
It seems that the hybrid model (\pophybrid) combines the advantages of 
the pure (\pop) model and both assignment models. 
A closer look at the single instances shows that it solved all instances
which are solved either by (\pop) or (\ass)+(c) or (\ass)+(e).
%
As indicated in our theoretical analysis
of the two models (POP) and (POP2), the hybrid model dominates the pure model for denser instances.
%
%
%
%
The behaviour of the two different versions of the (ASS) model shows a larger variation in running times: sometimes one of the versions is much faster than the others, while on other instances it is much slower.


\begin{table}[bt]
\centering
{
\small
\scalebox{0.63}{
  \setlength{\tabcolsep}{2.6pt} 
  \begin{tabular}{l|rr|l|c|ccc|ccc|ccc|ccc|ccc|c|c}
    \hline

	    &      &       &      & Time  &          &\rep  &      &      &\pop  &        &       &\pophybrid  &        &        &\ass+(c)  &        &        &\ass+(e)  &	&          &      \\
instance    &$|V|$ &$|E|$  &class & limit[s] &lb        &ub    &time  &lb    &ub    &time    &lb     &ub          &time    &lb      &ub        &time    &lb      &ub        &time    &old lb    &old ub\\
\hline
will199GPIA &701   &6772   &NP-s  & 3600 &-$\infty$ &660       &tl    &7     &7     &6.68	  &7      &7           &11.73   &7       &7         &6.81    &7       &7         &7.23	  &   &	   \\
ash331GPIA  &662   &4181   &NP-s  & 3600 &-$\infty$ &661       &tl    &4     &4     &11.54	  &4      &4           &15.26   &4       &4         &3.29    &4       &4         &75.06   &    &      \\   
ash608GPIA  &1216  &7844   &NP-m  & 3600 &-$\infty$ &1215      &tl    &4     &4     &34.84   &4      &4           &51.63   &4       &4         &575.23  &4       &4         &821.74  &   &\\
ash958GPIA  &1916  &12506  &NP-m  & 3600 &-$\infty$ &+$\infty$ &tl    &4     &4     &90.11   &4      &4           &105.77  &4       &6         &tl      &4       &8         &tl      &   &\\
abb313GPIA  &1555  &53356  &NP-?  & 7200 &-$\infty$ &853       &tl    &8     &10    &tl      &\textbf{9} &10      &tl      &8       &12        &tl      &8       &12	 &tl      &8\cite{MenZab06} &9\cite{Malaguti2011} \\ 
\hline


 
    \hline
  \end{tabular}
}
}
  \caption{Results of the simple models for the \texttt{GPIA} graphs from DIMACS benchmark set}
  \label{tab:resultsDIMACS:GPIA}
\end{table}

%

%

Table \ref{tab:resultsDIMACS:GPIA} shows the results 
for the simple models
for all five DIMACS  
\texttt{GPIA} graphs 
as well as
old lower and upper bounds for the instance 
\texttt{abb313GPIA}.
Since (\pop) and (\pophybrid) solved all 
\texttt{GPIA} graphs except 
\texttt{abb313GPIA} within a time limit of 1 h, we decided to increase the
time limit for this graph to 2 h. 
(\pophybrid) achieved a lower bound of 9 thus improving the best lower bound found. 
To our knowledge, this new model is the first one solving all of the remaining four \texttt{GPIA} graphs. 

\begin{figure}[b]
  \centering
  \begin{subfigure}[c]{0.49\textwidth}
    \includegraphics[scale=0.57]{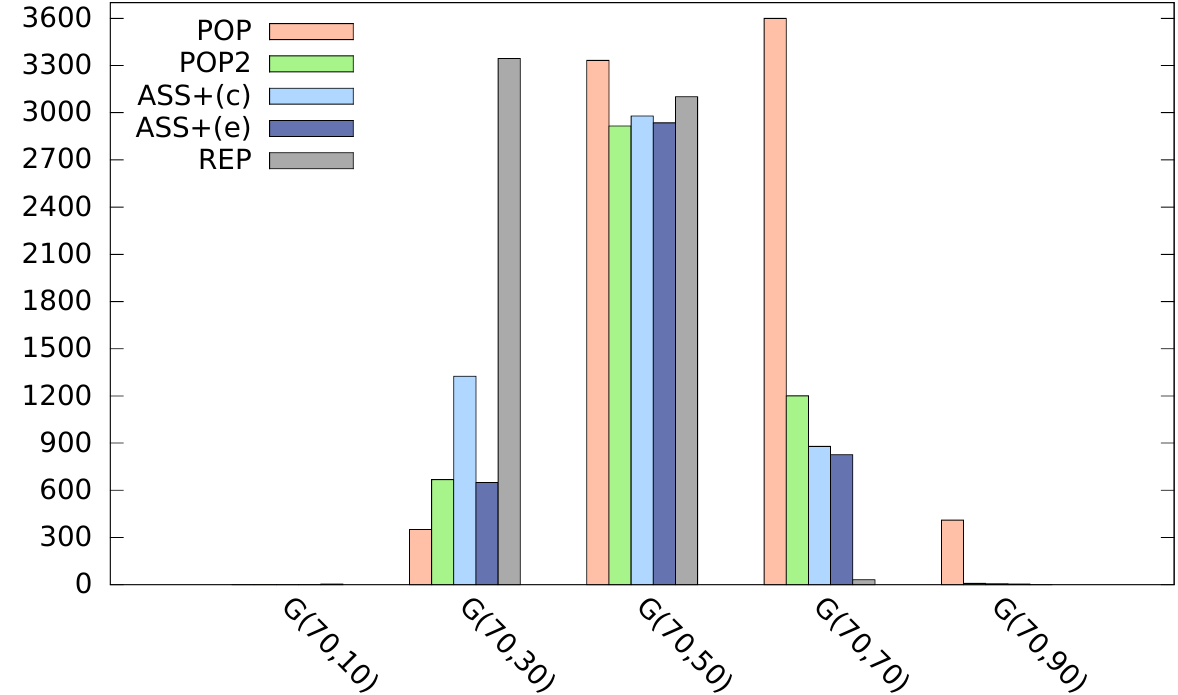}
    \subcaption{Average Execution Time {[s]}}
  \end{subfigure}
  \begin{subfigure}[c]{0.49\textwidth}
    \includegraphics[scale=0.57]{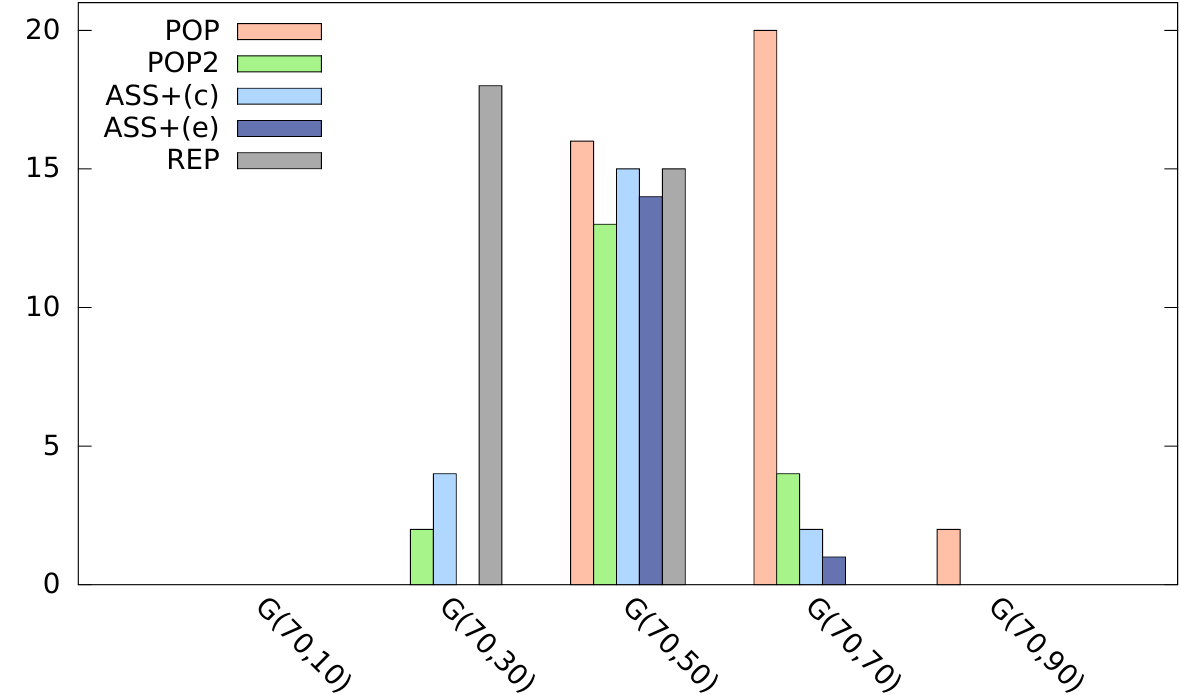}
    \subcaption{Number of unsolved instances}
  \end{subfigure}
  \caption{Comparison of the simple models
  on the benchmark \emph{set100} 
    }
  \label{fig:gnp_n70}
\end{figure}

\begin{figure}[bt]
\centering
    \begin{subfigure}[c]{0.49\textwidth}
        \includegraphics[scale=0.57]{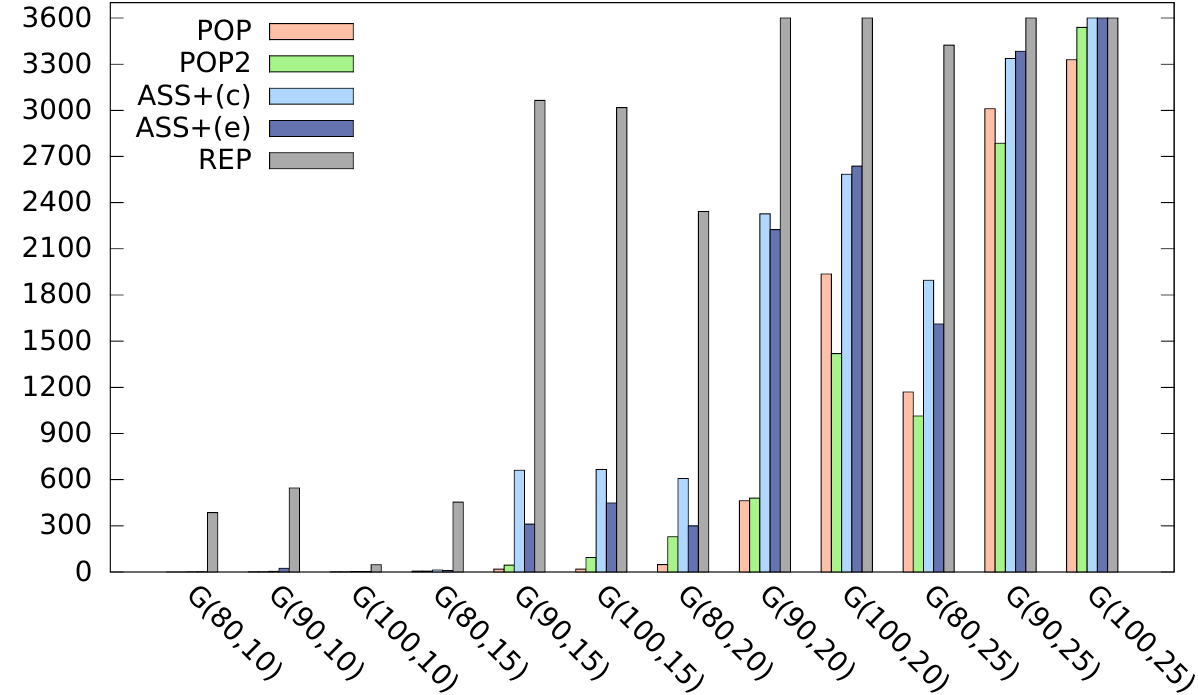}
   \subcaption{Average Execution Time {[s]}}
  \end{subfigure}
  \begin{subfigure}[c]{0.49\textwidth}
    \includegraphics[scale=0.57]{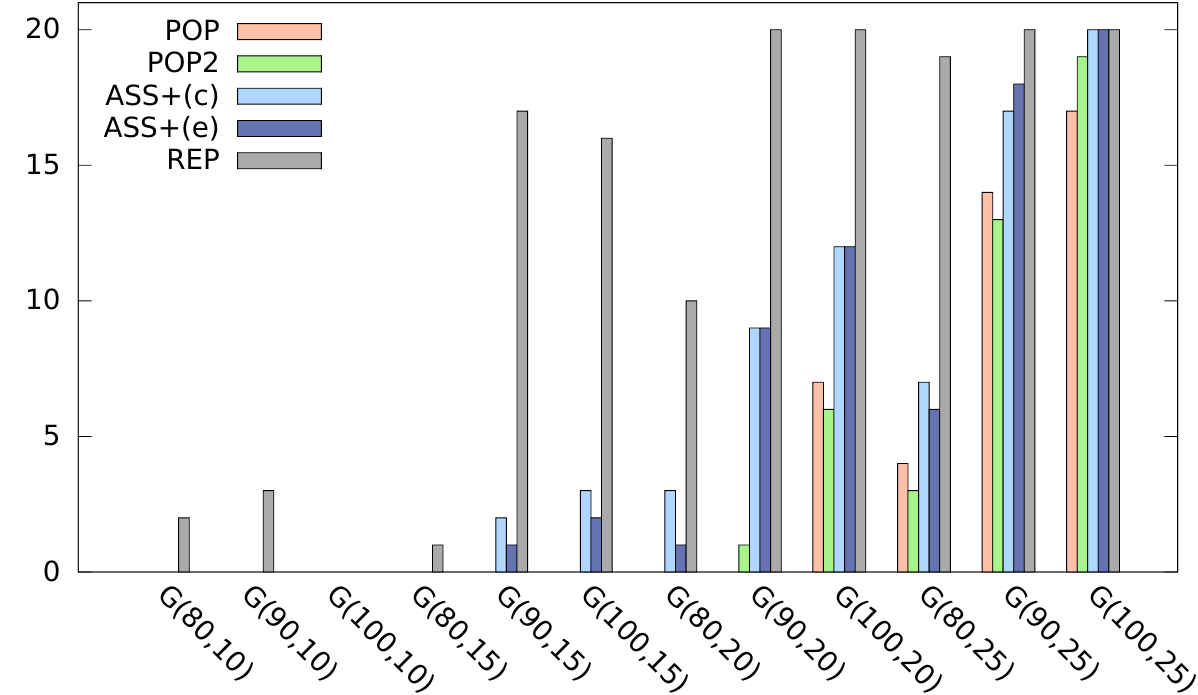}
    \subcaption{Number of unsolved instances}
  \end{subfigure}
  \caption{
    Comparison of the simple models
  on the benchmark \emph{sparse240}
  }
  \label{fig:gnp_s240}
\end{figure}
In order to study the behaviour of the implemented simple models for instances with varying size and density, 
we first
used the benchmark \emph{set100}, 
for which the results are displayed in Figure \ref{fig:gnp_n70}.
Figure \ref{fig:gnp_n70}(a) shows the average runtime for each set $G(70,p)$
for each density $p=10,30,50,70,90$, while 
Figure \ref{fig:gnp_n70}(b) shows the number of unsolved instances for each set.
(\pop) and (\ass)+(e) were able to solve all instances of densities 10 and 30,
where the average runtime of (\pop) is about 2 times shorter than that of (\ass)+(e).
From density 50 on, the (\pop) model seems to have more problems than the other models. 
The (\ass)+(c), (\ass)+(e) and the (\pophybrid) model deliver similar quality with similar running times for the denser graphs. From density 70 on, the (REP) model clearly dominates all other models, which was to be expected.
%
%
%
Since (\pop) was the fastest model on the sparse graphs of \emph{set100},
we decided to evaluate larger sparse graphs and generated the set \emph{sparse240}. 
The corresponding results are shown in Figure  \ref{fig:gnp_s240}. 
For the larger instances, both partial-ordering based models (POP) and (POP2) dominate the other models. 
The representative model already gets problems with the smallest problem instances in the set.

We conclude our work with answering our questions from the beginning:
\begin{itemize}
\item (H1): 
The model (POP2) is able to solve more DIMACS instances to provable optimality than both assignment models (ASS)+(c) and (ASS)+(e). This is not true for the (POP) model. On the random graphs the situation is similar: (POP2) clearly dominates both (ASS) models. 
(POP) dominates both (ASS) models only on the sparse instances. 
\item (H2): 
The (POP2) model dominates the original model (POP) 
on the harder DIMACS instances as well as on the tested dense random graphs.
The explanation lies in the fact discussed in section \ref{subsec:hybridMIP}. It seems that the (POP2) model combines the advantages of (POP) which is better for sparse graphs and (ASS) which is better for dense graphs.
\item (H3): 
The simple models are able to solve very hard instances from the DIMACS benchmark set. A comparison with the computational results of the state-of-the-art algorithms (such as \cite{MehTri96, Hansen2009, dimacs96,Sew96, Malaguti2011, Held2012, Malaguti2010, Segundo2012,MenZab08,MenZab06}) shows that the quality of the suggested algorithms is about the same (also see Table \ref{tab:resultsDIMACS} and \ref{tab:resultsDIMACS:GPIA}). Some of the approaches are able to solve some of the instances faster, but they are slower on other instances.
\item (H4): 
The representatives model does clearly dominate the other models on dense instances. This can be seen on the denser instances of the DIMACS graphs and on the series of random graphs with increasing density.
\end{itemize}








\bibliography{coloring-pop}


\end{document}